\documentclass
[superscriptaddress,pre,floats,final,nofootinbib,showkeys,twocolumn,10pt,letterpaper,lengthcheck]{revtex4}%
\usepackage{amssymb}
\usepackage{amsfonts}
\usepackage{amsmath}
\usepackage{graphicx}
\usepackage{epstopdf}
\usepackage{hyperref}%
\providecommand{\U}[1]{\protect\rule{.1in}{.1in}}
\newtheorem{theorem}{Theorem}

\newtheorem{lemma}{Lemma}

\newenvironment{proof}[1][Proof]{\noindent\textbf{#1.} }{\ \rule{0.5em}{0.5em}}
\begin{document}
\title[Competitive Epidemics]{May the Best Meme Win!: New Exploration of Competitive Epidemic Spreading over
Arbitrary Multi-Layer Networks}
\author{Faryad Darabi Sahneh}
\email{faryad@ksu.edu}
\affiliation{Electrical and Computer Engineering Department, Kansas State University}
\author{Caterina Scoglio}
\affiliation{Electrical and Computer Engineering Department, Kansas State University}
\begin{abstract}
This study extends the SIS\ epidemic model for single virus propagation over
an arbitrary graph to an SI$_{1}$SI$_{2}$S epidemic model of two exclusive,
competitive viruses over a two-layer network with generic structure, where
network layers represent the distinct transmission routes of the viruses. We
find analytical results determining extinction, mutual exclusion, and
coexistence of the viruses by introducing the concepts of survival threshold
and winning threshold. Furthermore, we show the possibility of coexistence in
SIS-type competitive spreading over multilayer networks. Not only do we
rigorously prove a region of coexistence, we quantitate it via interrelation
of central nodes across the network layers. Little to no overlapping of layers
central nodes is the key determinant of coexistence. Specifically, we show
coexistence is impossible if network layers are identical yet possible if the
network layers have distinct dominant eigenvectors and node degree vectors.
For example, we show both analytically and numerically that positive
correlation of network layers makes it difficult for a virus to survive while
in a network with negatively correlated layers survival is easier but total
removal of the other virus is more difficult. We believe our methodology has
great potentials for application to broader classes of multi-pathogen
spreading over multi-layer and interconnected networks.

\end{abstract}
\maketitle

\section{Introduction}

Multiple viral spreading within a single population involves very rich
dynamics \cite{Newman2011PRE}, attracting substantial attention
\cite{funk2010PRE, ahn2006PRE, wei2013JSAC}. Applications of these types of
models extend beyond physiological viruses, as `virus' may refer to products
\cite{aral2011MS}, memes \cite{Alex2012NSR}, pathogens \cite{Pej2011P}, etc.
Multiple virus propagation is a mathematically challenging problem. This
problem becomes particularly much more complicated if the network through
which viruses propagate are distinct. Current knowledge of how hybridity of
underlying topology influences fate of the pathogens is very little and
limited. These systems are usually mathematically intractable, hindering
conclusive results on spreading of multiple viruses on multi-layer networks.

Another source of complexity for this problem are multiple interaction
possibilities among viruses. For example, viruses may be reinforcing
\cite{newman2013X}, weakening \cite{granell2013X}, exclusive
\cite{newman2005PRL}, or asymmetric \cite{ahn2006PRE, wu2013JNS}. Newman
\cite{newman2005PRL} employed bound percolation to study the spread of two SIR
viruses in a host population through a single contact network, where a virus
takes over the network, then a second virus spreads through the resulting
residual network. The paper proved a coexistence threshold above the classical
epidemic threshold, indicating the possibility of coexistence in SIR model.
Karrer and Newman \cite{Newman2011PRE} extended the work to the more general
case where both viruses spread simultaneously. For SIS epidemic spreading,
Wang et el. \cite{wang2012NJP} studied competitive viruses and proved
exclusive, competitive SIS viruses cannot coexist in scale-free networks.

Multilayer networks generate interesting results for competitive viral
spreading. This type of models have implications in several applications like
product adoption (e.g. Apple vs. Android smart phones), virus-antidode
propagation, meme propagation, opposing opinions propagation, and etc. In
competitive spreading scenario, if infected by one virus, a node (individual)
cannot be infected by the other virus. Funk and Jansen \cite{funk2010PRE}
extended the bond percolation analysis of two competitive viruses to the case
of a two-layer network, investigating effects of layer overlapping. Granell et
al. \cite{granell2013X} studied the interplay between disease and information
co-propagation in a two-layer network consisting of one physical contact
network spreading the disease and a virtual overlay network propagating
information to stop the disease. They found a meta-critical point for the
epidemic onset leading to disease suppression. Importantly, this critical
point depends on awareness dynamics and the overlay network structure. Wei et
al. \cite{wei2012CCR} studied SIS spreading of two competitive viruses on an
arbitrary two-layer network, deriving sufficient conditions for exponential
die-out of both viruses. They introduced a statistical tool, EigenPredict, to
predict viral dominance of one competitive virus\ over the other
\cite{wei2013JSAC}.

In this paper, we address the problem of two competitive viruses propagating
in a host population where each virus has distinct contact network for
propagation. In particular, we study an $SI_{1}SI_{2}S$ model as the simplest
extension from SIS\ model for single virus propagation to competitive
spreading of two viruses on a two-layer network. From topology point of view,
our study is comprehensive because our multilayer network is allowed to have
any arbitrary structure.

Our paper is most relevant to \cite{wei2012CCR} and \cite{wei2013JSAC}. Wei et
al. conjectured in \cite{wei2012CCR} and numerically observed in
\cite{wei2013JSAC} that \textquotedblleft\emph{the meme whose first
eigenvalue}\footnote{Wei et al. \cite{wei2013JSAC} defined first eigenvalue of
of a meme as $\boldsymbol{\beta}\lambda_{1}-\boldsymbol{\delta}$, where
$\boldsymbol{\beta}$ is infection probabiltiy, $\boldsymbol{\delta}$ is curing
probabilty, and $\lambda_{1}$ is spectral radius of the underlying graph
layer.}\emph{ is larger tends to prevail eventually in the composite
networks.}\textquotedblright\ We challenge this argument from two aspects:
First, the definition of viral dominance in \cite{wei2013JSAC} is related to
comparison of fractions of nodes infected by each virus. However, when
comparing two viruses with two different contact networks, having a larger
eigenvalue is not a direct indicator of a higher final fraction of infected
nodes. In fact, it is possible to create two distinct network layers where a
meme spreading in the population with smaller eigenvalue takes over a much
larger fraction of the population. We find the definition of viral dominance
presented in \cite{wei2013JSAC} cannot be corroborated with eigenvalues
without severe restriction to a specific family of networks.

Second, and of paramount interest in this paper, largest eigenvalue is a graph
property\footnote{A graph property is any property on a graph which is
invariant under relabeling of nodes. Eigenvalues, degree moments, graph
diameter, etc. are examples of graph property.} of the layers in isolation and
thus does not have the capacity to discuss the joint influence of the network
topology, unless some sort of symmetry or homogeneity is assumed. In fact, the
generation of one layer in their synthetic multi-layer network via the Erdos
Reyni model \cite{wei2013JSAC} dictated a homogeneity in their multilayer
networks, creating a biased platform for further observations of layer
interrelations. Our work more accurately addresses network interrelation than
presented by Wei et al. \cite{wei2013JSAC} in moving beyond viral aggressivity
in isolation. We derived formulae more accurately and fully describing effect
of individual network layers and their interrelatedness.

We quantitate interrelations of contact layers in terms of spectral properties
of a set of matrices. Therefore, our results are not limited to any
homogeneity assumption or degree distribution and network model arguments. We
find analytical results determining extinction, mutual exclusion, and
coexistence of the viruses by introducing concepts of survival threshold and
winning threshold. Furthermore, we show possibility of coexistence in SIS-type
competitive spreading over multilayer networks. Not only do we prove a
coexistence region rigorously, we quantitate it via interrelation of central
nodes across the network layers. None or small overlapping of central nodes of
each layer is the key determinant of coexistence. We employ a novel multilayer
network generation framework to obtain a set of networks so that individual
layers have identical graph properties while the interrelation of network
layers varies. Therefore, any difference in outputs is purely the result of
interrelation. This makes ours a paradigmatic contribution to shed light on
topology hybridity in multilayer networks.

\section{Competitive Epidemics in Multi-Layer
Networks\label{Section: Modeling}}

In this paper, we study a continuous time $SI_{1}SI_{2}S$ model of two
competitive viruses propagating on a two-layer network, initially proposed in
discrete time\footnote{Wei et al. \cite{wei2012CCR} referred to their model as
$SI_{1}I_{2}S$. We prefer $SI_{1}SI_{2}S$ as a better candidate to emphasize
impossibility of direct transition between $I_{1}$ and $I_{2}$ in this
model.}\cite{wei2012CCR}.

\subsection{Multilayer Network Topology}

Consider a population of size $N$ among which two viruses propagate, acquiring
distinct transmission routes. Represented mathematically, the network topology
is a multi-layer network because two link types are present; one type allows
transmission of virus $1$ and the other type. allows transmission of virus
$2$. We represent this multilayer network as $\mathcal{G}(V,E_{A},E_{B})$,
where $V$ is the set of vertices (nodes) and $E_{A}$ and $E_{B}$ are set of
edges (links). By labeling vertices from $1$ to $N$, adjacency matrices
$A\triangleq\lbrack a_{ij}]_{N\times N}$ and $B\triangleq\lbrack
b_{ij}]_{N\times N}$ correspond to edge sets $E_{A}$ and $E_{B}$,
respectively, where $a_{ij}=1$ if node $j$ can transmit virus $1$ to node $i$,
otherwise $a_{ij}=0$ , and similarly $b_{ij}=1$ if node $j$ can transmit virus
$2$ to node $i$, otherwise $b_{ij}=0$. We assume the network layers are
symmetric, i.e., $a_{ij}=a_{ji}$ and $b_{ij}=b_{ji}$. Corresponding to
adjacency matrices $A$, we define $\boldsymbol{d}_{A}$ as the node degree
vector, i.e., $d_{A,i}=\sum_{j=1}^{N}a_{ij}$, $\lambda_{1}(A)$ as the largest
eigenvalue (or spectral radius) of $A,$ and $\boldsymbol{v}_{A}$ as the
normalized dominant eigenvector, i.e., $A\boldsymbol{v}_{A}=\lambda
_{1}(A)\boldsymbol{v}_{A}$ and $\boldsymbol{v}_{A}^{T}\boldsymbol{v}_{A}=1$.
We similarly define $\boldsymbol{d}_{A}$, $\lambda_{1}(A)$, and
$\boldsymbol{v}_{A}$ for adjacency matrix $B$.

Unlike simple, single-layer graphs, multilayer networks have not been studied
much in network science. We define simple graphs $G_{A}(V,E_{A})$ and
$G_{B}(V,E_{B})$ to refer to each isolated layer of the multilayer network
$\mathcal{G}(V,E_{A},E_{B})$. This allows us to argue multilayer network
$\mathcal{G}$\ in terms of simple graphs $G_{A}$ and $G_{B}$ properties and
their \emph{interrelation}. FIG. \ref{twolayernetwork.eps} shows a schematics
of the two-layer network.

\begin{figure}[ptb]
\centering
\includegraphics[ width=3 in]{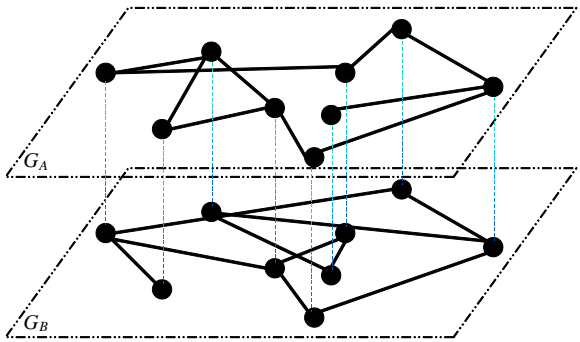}\caption{Schematics of
two-layer contact topology $\mathcal{G}(V,E_{A},E_{B})$, where a group of
nodes share two distinct interactions. In our $SI_{1}SI_{2}S$ model, virus $1$
transmits exclusively via $E_{A}$ links while virus $2$ transmits only through
$E_{B}$ links. Dotted vertical lines reiterate individual nodes are the same
in both layers of $\mathcal{G}$.}%
\label{twolayernetwork.eps}%
\end{figure}

\subsection{$SI_{1}SI_{2}S$ Model}

The $SI_{1}SI_{2}S$ model is an extension of continuous-time SIS spreading of
a single virus on a simple graph \cite{van2009TN, Ganesh2005INFOCOM} to
modeling of competitive viruses on a two-layer network. In this model, each
node is either `\emph{Susceptible},' `$I_{1}-$\emph{Infected},' or `$I_{2}%
-$\emph{Infected}' (i.e.,infected by virus $1$ or $2$, respectively), while
virus $1$ spreads through $E_{A}$ edges and virus $2$ spreads through $E_{B}$ edges.

In this competitive scenario the two viruses are exclusive: \emph{a node
cannot be infected by virus }$1$\emph{\ and virus }$2$\emph{\ simultaneously}.

Consistent with SIS\ propagation on a single graph (cf. \cite{van2009TN,
Ganesh2005INFOCOM}), the infection and curing processes for virus $1$ and $2$
are characterized by $(\beta_{1},\delta_{1})$ and $(\beta_{2},\delta_{2})$,
respectively. To illustrate, the curing process for $I_{1}-$infected node $i$
is a Poisson process with curing rate $\delta_{1}>0$. The infection process
for susceptible node $i$ effectively occurs at rate $\beta_{1}Y_{i}(t)$, where
$Y_{i}(t)$ is the number of $I_{1}-$infected neighbors of node $i$ at time $t$
in layer $G_{A}$. \emph{Effective infection rate} of a virus, defined as the
ratio of the infection rate over the curing rate, measures the expected number
of attempts of an infected node to infect its neighbor before recovering, thus
quantifying aggressiveness of a virus per contact. Curing and infection
processes for virus $2$ are similarly described. FIG.
\ref{schematicssi1i2s.eps}\ depicts a schematic of the $SI_{1}SI_{2}S$
competitive epidemic spreading model over a two-layer network.

\begin{figure}[ptb]
\centering
\includegraphics[ width=3.4 in]{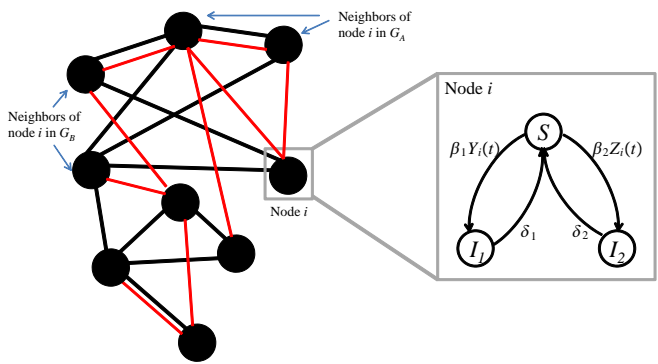}\caption{Schematics of a
contact network with the node-level stochastic transition diagram for node
$i$, according to the $SI_{1}SI_{2}S$ epidemic spreading model. Parameters
$\beta_{1}$ and $\delta_{1}$ denote virus $1$ infection rate and curing rate,
respectively, and $Y_{i}(t)$ is the number of node $i$ neighbors in layer
$G_{A}$ infected by virus $1$ at time $t$. Similarly, $\beta_{2}$ and
$\delta_{2}$ denote virus $2$ infection rate and curing rate, respectively,
and $Z_{i}(t)$ is the number of node $i$ neighbors in layer $G_{B}$ infected
by virus $2$ at time $t$.}%
\label{schematicssi1i2s.eps}%
\end{figure}

The $SI_{1}SI_{2}S$ model is essentially a coupled Markov process. For a
network with arbitrary structure, this model becomes mathematically
intractable due to exponential explosion of its Markov state space size
\cite{sahneh2013TON}. To overcome this issue with coupled Markov processes,
applying closure techniques results in approximate models with much smaller
state space size, however at the expense of accuracy. Specifically, a first
order mean-field type approximation \cite{sahneh2013TON} suggests the
following differential equations for the evolution of infection probabilities
of virus $1$ and $2$, denoted by $p_{1,i}$ and $p_{2,i}$ for node $i$,
respectively:%
\begin{align}
\dot{p}_{1,i}  &  =\beta_{1}(1-p_{1,i}-p_{2,i})\sum\nolimits_{j=1}^{N}%
a_{ij}p_{1,j}-\delta_{1}p_{1,i},\label{CSIS1}\\
\dot{p}_{2,i}  &  =\beta_{2}(1-p_{1,i}-p_{2,i})\sum\nolimits_{j=1}^{N}%
b_{ij}p_{2,j}-\delta_{2}p_{2,i}, \label{CSIS2}%
\end{align}
for $i\in\{1,...,N\}$, with the state-space size of $2N$. This model is an
extension of NIMFA model \cite{van2009TN} for SIS\ spreading on simple graphs.

Our competitive virus propagation model (\ref{CSIS1}-\ref{CSIS2}) exhibits
rich dynamical behavior dependent on epidemic parameters and contact network
multi-layer structure. Values of effective infection rates $\tau_{1}%
\triangleq\frac{\beta_{1}}{\delta_{1}}$ and $\tau_{2}\triangleq\frac{\beta
_{2}}{\delta_{2}}$\ of virus $1$ and $2$ yields several possible outcomes for
$SI_{1}SI_{2}S$ model (\ref{CSIS1}-\ref{CSIS2}). In particular, both viruses
may extinct ultimately, or one removes the other one, or both coexist.

\subsection{Problem Statement}

Linearization of our $SI_{1}SI_{2}S$ model (\ref{CSIS1}-\ref{CSIS2}) at the
healthy equilibrium (i.e. $p_{1,i}=p_{2,i}=0,i\in\{1,...,N\}$) demonstrates
the exponential extinction condition for both viruses. When $\tau
_{1}<1/\lambda_{1}(A)$ and $\tau_{2}<1/\lambda_{1}(B)$, any initial infections
exponentially die out. In this paper, we refer to such critical value as
\emph{no-spreading threshold} because a virus with a lower effective infection
rate is too weak to spread in the population even in the absence of any viral competition.

Wei et al. \cite{wei2012CCR} detailed the no-spreading condition as: If
$\tau_{1}<1/\lambda_{1}(A)$, virus $1$ does not spread, exponentially dying
out. Importantly, exponential extinction of both viruses occurs only if
$\tau_{1}<1/\lambda_{1}(A)$ and $\tau_{2}<1/\lambda_{1}(B)$ simultaneously.
Dynamical interplay between the competitive viruses does not affect the
no-spreading thresholds $\tau_{1}^{0}=1/\lambda_{1}(A)$ and $\tau_{2}%
^{0}=1/\lambda_{1}(B)$ for virus $1$ and virus $2$. These thresholds remain
independent of viral aggressivity of competitive viruses and network layers
interrelation. Exponential extinction is the only analytical outcome in Wei
\cite{wei2012CCR}. Our paper addresses two scenarios where for both viruses
$\tau_{1}>1/\lambda_{1}(A)$ and $\tau_{2}>1/\lambda_{1}(B)$.

\textbf{Problem:} Assume the effective infection rates of each virus is larger
than their no-spreading threshold, i.e., $\tau_{1}>1/\lambda_{1}(A)$ and
$\tau_{2}>1/\lambda_{1}(B)$:

1. Will both viruses survive (coexistence) or will one virus completely remove
the other (mutual exclusion)?

2. Which characteristics of multi-layer network structure allow for coexistence?

These questions pertain to long term behaviors of competitive spreading
dynamics. To address these questions, we perform a steady-state analysis of
$SI_{1}SI_{2}S$ model. Specifically, bifurcation techniques are used to find
two critical values: \emph{survival threshold} and \emph{winning threshold} to
determine if a virus will survive and whether it can completely remove the
other virus. Significantly, we go beyond these threshold conditions and
examine interrelation of network layers. Using eigenvalue perturbation, we
find interrelations of dominant eigenvectors and node-degree vectors of
network layers are critical determinants in ultimate behaviors of competitive
viral dynamics.

\section{Main Results\label{Section: Mainresult}}

Given our stated objective to study long-term behavior of $SI_{1}SI_{2}S$
model for competitive viruses, we use bifurcation analysis to study the
steady-state behavior of $SI_{1}SI_{2}S$ model. Application of bifurcation
analysis to the SIS\ model of a single virus on a simple graph determines the
critical value at which a non-healthy equilibrium emerges \cite{van2009TN},
determining a survival threshold for the virus. Interestingly, no-spreading
threshold and survival threshold coincide for this SIS\ model. However, we
expect these two critical values are distinct for $SI_{1}SI_{2}S$ because a
virus may initially spread in an almost entirely susceptible population but
then die out from competition with a simultaneous virus having a sufficiently
stronger infection rate.

In fact, the survival threshold is larger than the no-spreading threshold,
monotonically increasing with the aggressivity of the other competitive virus.
Furthermore, a surviving virus can even be so aggressive to completely remove
the other virus. Consequently, competitive spreading induces an additional
threshold concept-the winning threshold-determining the critical value of
effective infection rate for a virus to prevail as sole survivor.

The determination of the two thresholds for each virus involves four
quantities. We are able to deduce winning thresholds from survival thresholds,
which then become our sole focus. Furthermore, with no loss of generality, we
only find survival threshold of virus $1$ because of expressions duality.

Unfortunately, any conclusive understanding of the system is hindered by the
complex interdependency of survival threshold of one virus on the multilayer
network topology and the aggressiveness of the competitive virus. While
complete analytical solution of survival threshold appears impossible, we
characterize possible solutions with explicit analytical expressions. This
step is a unique contribution to current understanding of competitive
spreading over multi-layer networks with solid and quantitative implications
on role of multilayer network topology.

\subsection{Threshold Equations}

Bifurcation analysis of $SI_{1}SI_{2}S$\ model equilibriums finds the survival
threshold. Our competitive virus propagation model (\ref{CSIS1}-\ref{CSIS2})
yields the equilibriums equations:%
\begin{align}
\frac{p_{1,i}^{\ast}}{1-p_{1,i}^{\ast}-p_{2,i}^{\ast}}  &  =\tau_{1}\sum
a_{ij}p_{1,j}^{\ast},\label{EqEq1}\\
\frac{p_{2,i}^{\ast}}{1-p_{1,i}^{\ast}-p_{2,i}^{\ast}}  &  =\tau_{2}\sum
b_{ij}p_{2,j}^{\ast}, \label{EqEq2}%
\end{align}
for $i\in\{1,...,N\}$. The healthy equilibrium (i.e., $p_{1,i}^{\ast}%
=p_{2,i}^{\ast}=0,\forall i$) is always a solution to the above equilibrium
equation (\ref{EqEq1}-\ref{EqEq2}). Long term persistence of infection in the
population is associated with non-zero solution for the equilibrium equations
\cite{van2009TN}. We use bifurcation theory to identify critical values for
effective infection rates $\tau_{1}$ and $\tau_{2}$ such that a second
equilibrium, aside from the healthy equilibrium, emerges. The critical value
for one virus is a function of the effective infection rate of the other
virus. Without loss of generality, we determine the survival threshold for
virus $1$ by finding the critical effective infection rate $\tau_{1c}$ as a
function of $\tau_{2}$.

\textbf{Definition:} Given virus $2$ effective infection rate ($\tau_{2}$),
the \emph{survival threshold value} $\tau_{1c}$ is the smallest effective
infection rate that virus $1$ steady state infection probability of each node
is positive for $\tau_{1}>\tau_{1c}$. For $\tau_{2}$ in $[0,+\infty)$ as an
independent variable, $\tau_{1c}$ constitutes a \emph{survival threshold
curve}, monotonically increasing function of $\tau_{2}$, denoted by $\Phi
_{1}(\tau_{2})$.

The above definition for survival threshold value indicates that exactly at
the threshold value, $p_{1,i}^{\ast}|_{\tau_{1}=\tau_{1c}}=0$ and
$\frac{dp_{1,i}^{\ast}}{d\tau_{1}}|_{\tau_{1}=\tau_{1c}}>0$ for all
$i\in\{1,...,N\}$. Taking the derivative of equilibrium equations
(\ref{EqEq1}) with respect to $\tau_{1}$, and defining
\begin{equation}
w_{i}\triangleq\frac{dp_{1,i}^{\ast}}{d\tau_{1}}|_{\tau_{1}=\tau_{1c}}%
,~y_{i}\triangleq p_{2,i}^{\ast}|_{\tau_{1}=\tau_{1c}},
\end{equation}
we find the survival threshold $\tau_{1c}$ is the value for which nontrivial
solution exists for $w_{i}>0$ in%
\begin{equation}
w_{i}=\tau_{1c}(1-y_{i})\sum a_{ij}w_{j}, \label{Thresh_Eq}%
\end{equation}
where $y_{i}$ is the solution of:%
\begin{equation}
\frac{y_{i}}{1-y_{i}}=\tau_{2}\sum b_{ij}y_{j}, \label{yi_eq}%
\end{equation}
according to equilibrium equation (\ref{EqEq2}).

Equation (\ref{Thresh_Eq}) is an eigenvalue problem. Among all the possible
solutions, only%
\begin{equation}
\tau_{1c}=\frac{1}{^{\lambda_{1}(diag\{1-y_{i}\}A)}} \label{tawc1_yi}%
\end{equation}
is acceptable; according to Perron-Frobenius Theorem, only the dominant
eigenvector of the matrix $diag\{1-y_{i}\}A$ has all positive entries,
yielding $w_{i}=\frac{dp_{1,i}^{\ast}}{d\tau_{1}}|_{\tau_{1}=\tau_{1c}}>0$.

The eigenvalue problem (\ref{Thresh_Eq}) gives a mathematical way to find the
survival threshold $\tau_{1c}$, depending on the value of $\tau_{2}$.
Unfortunately, this implicit dependence hinders clear understanding of the
propagation interplay between virus $1$ and virus $2$.

Finding $y_{i}$ for all possible values of $\tau_{2}$, then finding the
threshold value $\tau_{c1}$ from (\ref{tawc1_yi}), we obtain survival
threshold curve $\Phi_{1}(\tau_{2})$ for virus $1$. This curve divides the
region of $(\tau_{1},\tau_{2})$ into two regions, where virus $1$ survives and
one where virus $1$ extincts. We can use analogous equations to find survival
threshold curve $\Phi_{2}(\tau_{1})$ for virus $2$. Given $\tau_{2}$, we can
find $\tau_{1c}$ such that for $\tau_{1}>\tau_{1c}$, virus $1$ can survive.

We can think of another threshold, winning threshold, such that for $\tau
_{1}>\tau_{1}^{\dag}$, only virus $1$ can survive and virus $2$ is completely
suppressed. Interestingly, the winning threshold of virus $1$ is the value of
$\tau_{1}^{\dag}$, such that the survival threshold of virus $2$ is $\tau_{2}$
for $\tau_{1}=\tau_{1}^{\dag}$. Therefore, $\Psi_{1}(\cdot)$ is the inverse
function of $\Phi_{2}(\cdot)$, i.e.,%
\begin{equation}
\Psi_{1}(\tau_{2})=\Phi_{2}^{-1}(\tau_{2}). \label{inverse_threshold}%
\end{equation}
Therefore, finding the survival thresholds of both viruses also yields the
winning threshold curves. The two curves $\Phi_{1}(\tau_{2})$ and $\Phi
_{2}(\tau_{1})$ divide $(\tau_{1},\tau_{2})$ plane in four regions: where both
viruses extinct, where only virus $1$ survives, where only virus $2$ survives,
where both viruses survive and coexist. The coexisting region contains the
values of $(\tau_{1},\tau_{2})$ between survival threshold curves $\Phi
_{1}(\tau_{2})$ and $\Phi_{2}(\tau_{1})$.

\subsection{Characterization of Threshold Curves}

Complete analytical solution of survival threshold curves is not feasible.
Instead, we quantitate interrelations of contact layers to formulate our
analytical assertions. We describe conditions for viral coexistence through
attaining explicit analytical quantities giving conditions for mutual
exclusion and coexistence of viruses. Our approach to this problem finds
explicit solutions to (\ref{Thresh_Eq}) and (\ref{yi_eq}) for values of
$\tau_{2}$ close to $1/\lambda_{1}(B)$ and for very large values of $\tau_{2}$
to\ quantitate the survival epidemic curves. Since we know solution to
(\ref{yi_eq}) and the survival threshold value $\tau_{1c}$ at both extreme
values, we can employ eigenvalue perturbation techniques to find explicit
solutions for $\tau_{2}$ close to $1/\lambda_{1}(B)$ and $\tau_{2}$ very
large. Results for $\tau_{2}$ close to $1/\lambda_{1}(B)$ apply where
competitive viruses are \emph{non-aggressive}, whereas results for $\tau_{2}$
very large corresponds to \emph{aggressive} competition. Behavior of the
competitive spreading process for moderate aggressiveness is an interpolation
of the extreme scenarios of non-aggressive and aggressive propagation.

First, we perform perturbation analysis to find $\tau_{c1}$ for values of
$\tau_{2}$ close to $1/\lambda_{1}(B)$. We know at $\tau_{2}=1/\lambda_{1}%
(B)$, $y_{i}=0$ solves (\ref{yi_eq}), thus $\tau_{c1}=1/\lambda_{1}(A)$ is the
survival threshold according to (\ref{yi_eq}). For values of $\tau_{2}$ close
to $1/\lambda_{1}(B)$, we use eigenvalue perturbation technique and study
sensitivity of threshold equation (\ref{Thresh_Eq}) respective to deviation in
$\tau_{2}$ from $1/\lambda_{1}(B)$. As detailed in the Appendix, we find%
\begin{equation}
\frac{d\tau_{1c}}{d\tau_{2}}|_{\tau_{2}=\frac{1}{\lambda_{1}(B)}}%
=\frac{\lambda_{1}(B)}{\lambda_{1}(A)}\frac{\sum v_{A,i}^{2}v_{B,i}}{\sum
v_{B,i}^{3}},\label{dtaw1cdtaw2small}%
\end{equation}
expressing the dependency of virus $1$ survival threshold ($\tau_{1c}$) to
effective infection rate of virus $2$ ($\tau_{2}$) for values of $\tau_{2}$
close to $1/\lambda_{1}(B)$. Expression (\ref{dtaw1cdtaw2small}) consists of
two components: $\frac{\lambda_{1}(B)}{\lambda_{1}(A)}$, the spectral radius
ratio of each network layers in isolation, and $\frac{\sum v_{A,i}^{2}v_{B,i}%
}{\sum v_{B,i}^{3}}$, which determines the influence of interrelations of the
two layers. Significantly, if $\sum v_{A,i}^{2}v_{B,i}$ is small, expression
(\ref{dtaw1cdtaw2small}) suggests the virus $1$ survival threshold is not
influenced by virus $2$ infection rate. This has very interesting
interpretations: when spectral central nodes of $G_{A}$ (those nodes with
larger element in dominant eigenvector of $G_{A}$) are are spectrally
insignificant in $G_{B}$, the virus $1$ survival threshold does not increase
much by $\tau_{2}$. In other words, virus $2$ does not compete over accessible
resources of virus $1$, therefore, virus $1$ is not affected much by the
co-propagation. On the other hand, if spectral central nodes of $G_{A}$ have
high spectral centrality in $G_{B}$, then $\sum v_{A,i}^{2}v_{B,i}$ is maximal
indicating considerable dependency of survival threshold of virus $1$ on
aggressiveness of the other virus. From (\ref{dtaw1cdtaw2small}), the die-out
threshold curve $\Phi_{1}(\tau_{2})$ can be approximated close to $(\tau
_{2},\tau_{1})=(\frac{1}{\lambda_{1}(B)},\frac{1}{\lambda_{1}(A)})$ as%
\begin{equation}
\Phi_{1}(\tau_{2})\simeq\frac{1}{\lambda_{1}(A)}\{1+\frac{\sum v_{A,i}%
^{2}v_{B,i}}{\sum v_{B,i}^{3}}(\lambda_{1}(B)\tau_{2}-1)\}.
\end{equation}

\begin{figure}[ptb]
\centering
\includegraphics[ width=3.4 in]{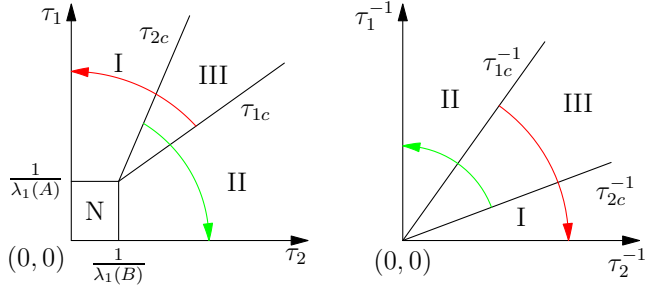}\caption{The survival
regions diagram in $SI_{1}SI_{2}S$ model for values of $(\tau_{1},\tau_{2})$
close to $(\frac{1}{\lambda_{1}(A)},\frac{1}{\lambda_{1}(B)})$ (left) and for
very large values of $(\tau_{1},\tau_{2})$ (right). The $SI_{1}SI_{2}S$ model
with two-layer contact topology exhibits four possibilities: extinction region
N where both viruses die-out, mutual extinction region I, where virus $1$
survives and virus $2$ dies out, mutual extinction region II, where only virus
2 survives and virus 1 dies out, and finally coexistence region III, where
both viruses survive and persist in the population. The red arrow shows the
survival region of virus $1$ (regions I and III) and the green arrow shows the
survival region of virus $2$ (regions II and III). For aggressive viruses
scenario, axes have inversed values of $(\tau_{1},\tau_{2})$ so that the
origin represents infinitely large values. Equations (\ref{dtaw1cdtaw2small})
and (\ref{taw1ctaw2large}) analytically find the separating lines between the
survival regions in explicit expressions.}%
\label{regionsextreme.eps}%
\end{figure}

Studying threshold equations (\ref{Thresh_Eq})-(\ref{yi_eq}) for $\tau
_{2}\rightarrow\infty$, we find $\frac{\tau_{1c}}{\tau_{2}}|_{\tau
_{2}\rightarrow\infty}$ is the inverse of the spectral radius of $D_{B}^{-1}A$
(see Appendix for detailed derivation):%
\begin{equation}
\frac{\tau_{1c}}{\tau_{2}}|_{\tau_{2}\rightarrow\infty}=\frac{1}{\lambda
_{1}(D_{B}^{-1}A)}=\frac{1}{\lambda_{1}(D_{B}^{-1/2}AD_{B}^{-1/2}%
)},\label{taw1ctaw2large}%
\end{equation}
expressing the dependency of virus $1$ survival threshold ($\tau_{1c}$) on
effective infection rate of virus $2$ ($\tau_{2}$) for large values of
$\tau_{2}$. This expression (\ref{taw1ctaw2large}) directly highlights the
influence of interrelations of the two layers. Significantly, if $\lambda
_{1}(D_{B}^{-1}A)$ is large, expression (\ref{dtaw1cdtaw2small}) suggests that
virus $1$ survival threshold does not increase significantly by virus $2$
infection rate. Similar arguments about interpretation of
(\ref{dtaw1cdtaw2small}) apply to aggressive competitive viruses where
$\tau_{1}$ and $\tau_{2}$ are relatively large. The main difference in case of
aggressive competitive spreading is that node degree is the determinant of
centrality. From (\ref{taw1ctaw2large}), the die-out threshold curve $\Phi
_{1}(\tau_{2})$ asymptotically becomes%
\begin{equation}
\Phi_{1}(\tau_{2})\simeq\frac{1}{\lambda_{1}(D_{B}^{-1}A)}\tau_{2},
\end{equation}
for aggressive competitive propagation. FIG. \ref{regionsextreme.eps} depicts
survival threshold curves for non-aggressive (left) and aggressive (right)
competitive spreading.

We prove conditions for coexistence by showing there is overlapping between
regions where viruses survive.

\begin{theorem}
In $SI_{1}SI_{2}S$ model (\ref{CSIS1}-\ref{CSIS2}) for competitive epidemics
over multi-layer networks, if the two network layers $G_{A}$ and $G_{B}$ are
identical, coexistence is impossible, i.e., a virus with even a slightly
larger effective infection rate dominates and completely removes the other
virus. Otherwise, if node-degree vectors of $G_{A}$ and $G_{B}$ are not
parallel, i.e., $\boldsymbol{d}_{A}\not =c\boldsymbol{d}_{B}$, or dominant
eigenvectors of $G_{A}$ and $G_{B}$ do not completely overlap, i.e.,
$\boldsymbol{v}_{A}\not =\boldsymbol{v}_{B}$ the multi-layer structure of the
underlying topology allows a nontrivial coexistence region.
\end{theorem}

\begin{proof}
If $G_{A}=G_{B}$, then equation (\ref{yi_eq}) suggests $\tau_{c1}=\tau_{2}$
solves threshold equation (\ref{Thresh_Eq}). Similarly $\tau_{2c}=\tau_{1}$,
suggesting $\tau_{1}^{\dag}=\tau_{2}$ according to (\ref{inverse_threshold}),
i.e., survival and winning thresholds coincide. Therefore, the virus with even
a slightly larger effective infection rate dominates and completely removes
the other virus if the two network layers are identical.

In order to show possibility of coexistence for non-aggressive competitive
viruses, we show the survival regions overlap by proving%
\begin{equation}
\frac{d\tau_{1,c}}{d\tau_{2}}.\frac{d\tau_{2,c}}{d\tau_{1}}|_{(\tau_{1}%
,\tau_{2})=(\frac{1}{\lambda_{1}(A)},\frac{1}{\lambda_{1}(B)})}<1.
\label{coexistence_1}%
\end{equation}
Using expression (\ref{dtaw1cdtaw2small}) and its counterpart for $\frac
{d\tau_{2,c}}{d\tau_{1}}$ (see Appendix), we find condition
(\ref{coexistence_1}) is always true except for the special case where
dominant eigenvectors of $G_{A}$ and $G_{B}$ completely overlap, i.e.,
$\boldsymbol{v}_{A}=\boldsymbol{v}_{B}$.

In order to show possibility of coexistence for aggressive competitive
viruses, we show the survival regions overlap by proving%
\begin{equation}
\frac{\tau_{1c}}{\tau_{2}}|_{\tau_{2}\rightarrow\infty}\times\frac{\tau_{2c}%
}{\tau_{1}}|_{\tau_{2}\rightarrow\infty}<1. \label{coexistence_2}%
\end{equation}

Using expression (\ref{taw1ctaw2large}) and its counterpart for $\frac
{\tau_{2c}}{\tau_{1}}|_{\tau_{2}\rightarrow\infty}$ (see Appendix), we find
that condition (\ref{coexistence_2}) is always true except for the special
case where node-degree vectors of $G_{A}$ and $G_{B}$ are parallel, i.e.,
$\boldsymbol{d}_{A}=c\boldsymbol{d}_{B}$.
\end{proof}

When dominant eigenvectors of $G_{A}$ and $G_{B}$ are not identical, condition
(\ref{coexistence_1}) indicates non-aggressive viruses can coexist. When
propagation of competitive viruses is aggressive, condition
(\ref{coexistence_2}) indicates viruses can coexist if node-degree vectors of
$G_{A}$ and $G_{B}$ are not parallel. However, the rare scenario where $G_{A}$
and $G_{B}$ are not identical and $\boldsymbol{d}_{A}=c\boldsymbol{d}_{B}$ and
$\boldsymbol{v}_{A}=\boldsymbol{v}_{B}$ hold simultaneously demands further exploration.

The above theorem and equations (\ref{dtaw1cdtaw2small}) and
(\ref{taw1ctaw2large}) prove the importance of interrelation of network
layers. As will be discussed in the simulation section, one approach capturing
only the effect of interrelation is generating multilayer networks from two
graphs $G_{A}$ and $G_{B}$ through simple relabeling vertices of $G_{B}$. We
thus have a set of multilayer networks whose layers have identical graph
properties but correpondence of nodes in one layer to the nodes of the other varies.

In the context of competitive spreading, whether memes, opinions, or products,
the population under study serves as the `resource' for the competitive
entities, relating nicely to the concept of `competing species' in ecology.
Longterm study of competing species in ecology centers on the `competitive
exclusion principle' \cite{hardin1960S}: \emph{Two species competing for the
same resources cannot coexist indefinitely under identical ecological factors.
The species with the slightest advantage or edge over another will dominate
eventually}. Our \textit{SI}$_{1}$\textit{SI}$_{2}$\textit{S} model also
predicts when the network layers are identical, coexistence is not possible.
Significantly, different propagation routes break this `ecological symmetry,'
allowing coexistence. Not only have we rigorously proved a coexistence region,
we quantitated this ecological asymmetry via interrelation of central nodes
across the network layers. None or small overlapping of central nodes of each
layer is the key determinant of coexistence. Excitingly, this conclusion
nicely relates to `niche differentiation' in ecology and yet is built upon
network science rigor.

\subsection{Standardized Threshold Diagram and a Global Approximate Formula}

Exploring efficient characterization of threshold curves using extreme
scenarios, we propose a standardized threshold diagram, where threshold curves
are plotted in a $[0,1]\times\lbrack0,1]$ plane for $(x,y)=(\frac{1}%
{\lambda_{1}(B)\tau_{2}},\frac{1}{\lambda_{1}(A)\tau_{1}})$, axes scaled by
layer spectral radius and inverted. Curves in standardized threshold diagram
start from origin to point $(1,1)$. From (\ref{dtaw1cdtaw2small}) and
(\ref{taw1ctaw2large}) the slopes of the survival curve of virus $1$ at
$(0,0)$ and $(1,1)$ are%
\begin{align}
m_{0} &  =\frac{\lambda_{1}(B)}{\lambda_{1}(A)}\lambda_{1}(D_{B}%
^{-1}A),\label{m0}\\
m_{1} &  =\frac{\sum v_{A,i}^{2}v_{B,i}}{\sum v_{B,i}^{3}},\label{m1}%
\end{align}
respectively. Importantly, these slopes help creating a parametric
approximation for the survival threshold curve $\tau_{1c}=\Phi_{1}(\tau_{2})$
for the full range of $\tau_{2}$. We use a quadratic Bezier curve
\begin{equation}%
\begin{bmatrix}
x\\
y
\end{bmatrix}
=2\sigma(1-\sigma)%
\begin{bmatrix}
a\\
b
\end{bmatrix}
+\sigma^{2}%
\begin{bmatrix}
1\\
1
\end{bmatrix}
,\label{Bezier}%
\end{equation}
connecting $(x,y)=(0,0)$ to $(x,y)=(1,1)$ for $\sigma\in\lbrack0,1]$, and
satisfying the slope constraints (\ref{m0}) and (\ref{m1}), if $a$ and $b$ are
chosen as:
\begin{equation}
a=\frac{1-m_{1}}{m_{0}-m_{1}},b=\frac{m_{0}(1-m_{1})}{m_{0}-m_{1}}.
\end{equation}

Therefore, the Bezier curve (\ref{Bezier}) approximates the standardized
threshold curve diagram for the whole range of $\tau_{1}>1/\lambda_{1}(A)$ and
$\tau_{2}>1/\lambda_{1}(B)$ using only spectral information of a set of matrices.

\subsection{Multi-layer Network Metric for Competitive Spreading}

Proving coexistence is one of the key contributions of this paper. We go
further to define a topological index $\Gamma_{s}(\mathcal{G})$ quantifying
possibility of coexistence in a multi-layer network $\mathcal{G=(}%
V,E_{A},E_{B}\mathcal{)}$ for the case of non-aggressive spreading as%
\[
\Gamma_{s}(\mathcal{G})=1-\frac{(\sum v_{B,i}v_{A,i}^{2})(\sum v_{A,i}%
v_{B,i}^{2})}{(\sum v_{B,i}^{3})(\sum v_{A,i}^{3})}.
\]
Values of $\Gamma_{s}(\mathcal{G})$ vary from $0$ (corresponding to the case
where $\boldsymbol{v}_{A}=\boldsymbol{v}_{B}$) to $1$. Values of $\Gamma
_{s}(\mathcal{G})$ close to zero imply coexistence is rare and any survived
virus is indeed the absolute winner. $\Gamma_{s}(\mathcal{G})$ closer to $1$
indicates coexistence is very possible on $\mathcal{G}$. Therefore,
$\Gamma_{s}(\mathcal{G})$ can be used to discuss coexistence of non-aggressive
competitive viruses.

Similar to non-aggressive competitive spreading, we can define a topological
index $\Gamma_{l}(\mathcal{G})$ to quantify coexistence possibility in a
multi-layer network $\mathcal{G=(}V,E_{A},E_{B}\mathcal{)}$ as%
\[
\Gamma_{l}(\mathcal{G})=1-\frac{1}{\lambda_{1}(D_{B}^{-1}A)}.\frac{1}%
{\lambda_{1}(D_{A}^{-1}B)}.
\]

Values of $\Gamma_{l}(\mathcal{G})$ vary from $0$ (corresponding to the case
where $\boldsymbol{d}_{A}=c\boldsymbol{d}_{B}$) to $1$. Values of $\Gamma
_{l}(\mathcal{G})$ close to zero imply coexistence is rare and any survived
virus is indeed the absolute winner. $\Gamma_{l}(\mathcal{G})$ closer to $1$
indicates coexistence is very possible on $\mathcal{G}$. Therefore,
$\Gamma_{l}(\mathcal{G})$ can be used to discuss coexistence of aggressive
competitive viruses.

\subsection{Numerical Simulations}

\textbf{Multi-layer network generation: }Our objective for numerical
simulations is not only to test our analytical formulae, but also to
investigate our prediction of cross-layer interrelation effect on competitive
epidemics. This demands a set of two-layer networks for which isolated layers
have identical graph properties but how these layers are interrelated is
different, hence capturing the \emph{pure effect of interrelation}.
Specifically, in the following numerical simulations, the contact network
$G_{A}$ through which virus $1$ propagates is a random geometric graph with
$N=1000$ nodes, where pairs less than $r_{c}=\sqrt{\frac{3\log(N)}{\pi N}}$
apart connect to ensure connectivity. For the contact graph of virus $2$
($G_{B}$), we first generated a scale-free network according to the
Barab\'{a}si--Albert model. We then used a randomized greedy algorithm to
associate the nodes of this graph with the nodes of $G_{A}$, approaching a
certain degree correlation coefficient $\rho$ with $G_{A}$, i.e., each
iteration step permutates nodes when the degree correlation coefficient%
\[
\rho(\mathcal{G})=\frac{\sum(d_{A,i}-\bar{d}_{A})(d_{B,i}-\bar{d}_{B})}%
{\sqrt{\sum(d_{A,i}-\bar{d}_{A})^{2}}\sqrt{\sum(d_{B,i}-\bar{d}_{B})^{2}}},
\]
is closer to the desired value. Specifically, we obtained three different
permutations where the generated graphs are negatively ($\rho=-0.47$),
neutrally ($\rho=0$), and positively ($\rho=0.48$) correlated with $G_{A}$.
These three graphs have \textit{identical graph properties, yet they are
distinct respective to }$G_{A}$. FIG. \ref{networkperm.eps} depicts a graph
$G_{A}$ and three graphs of $G_{B}$ with $N=100$ nodes to improve conceptualization.

\begin{figure}[ptb]
\centering
\includegraphics[ width=3.4 in]{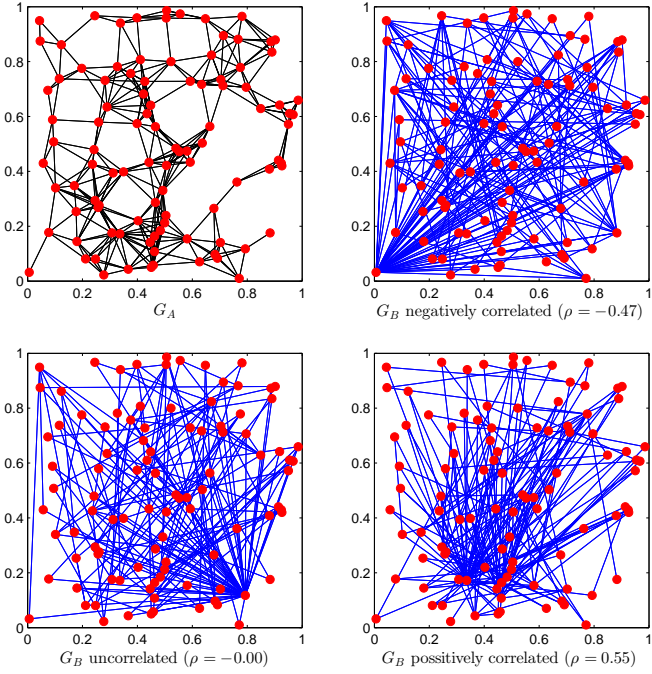}\caption{\emph{Two-layer
network generation for numerical simulations is generated here.} The contact
network $G_{A}$ through which virus $1$ propagates is a random geometric graph
where pairs of nodes with a distance less than $r_{c}$ are connected to each
other. For visualization convenience, the number of nodes is $N=100$, which is
different from the actual $N=1000$ used for numerical simulation results. For
the contact graph of virus $2$ ($G_{B}$), we first generated a scale-free
network according to the B-A model, associating the nodes of this graph with
the nodes of $G_{A}$ to achieve a certain degree correlation coefficient with
$G_{A}$. Specifically, we obtained three different permutations such that the
generated graphs are negatively, neutrally, and positively correlated with
$G_{A}$. These three graphs are the same if isolate, and distinct in their
interrelation with $G_{A}$. The high degree nodes in the positively correlated
$G_{B}$ (lower right) have also high degree in $G_{A}$ (upper left), while the
high degree nodes in the negatively correlated $G_{B}$ (upper right) have low
degree size in $G_{A}$. The uncorrelated $G_{B}$ (lower left) shows no clear
association.}%
\label{networkperm.eps}%
\end{figure}

\textbf{Steady-state infection fraction:} When the spreading of a single virus
is modeled as SIS, the steady-state infection fraction $\bar{p}^{ss}=\frac
{1}{N}\sum p_{i}$ illustrates a threshold phenomena respective to effective
infection rates: steady-state infection fraction $\bar{p}^{ss}$ is zero for
effective infection rates less than a critical value but becomes positive for
larger values. When two viruses compete to spread, steady state infection
fraction $\bar{p}_{1}^{ss}=\frac{1}{N}\sum p_{1,i}$ of virus $1$ in the
$SI_{1}SI_{2}S$ model exhibits a threshold behavior at $\tau_{1}=\tau_{1c}$,
for a given $\tau_{2}$. FIG. \ref{ssinfectionpositive.eps} depicts the steady
state infection fraction curve of virus $1$ in the $SI_{1}SI_{2}S$ competitive
spreading model. In this simulation, effective infection rate of virus $2$ is
fixed at $\tau_{2}=6\frac{1}{\lambda_{1}(B)}$ and $G_{B}$ is positively
correlated with $G_{A}$ ($\rho=0.48$). In order to obtain a unified form, we
normalized the horizontal axis to $\tau_{1}\lambda_{1}(A)$. The steady state
infection fraction of virus $1$, $\bar{p}_{1}^{ss}$, is zero for $\tau_{1}%
\leq\tau_{1c}\simeq3\frac{1}{\lambda_{1}(A)}$, identifying this range as an
extinction region for virus $1$, while $\bar{p}_{1}^{ss}$ is positive for
$\tau_{1}>\tau_{1c}$ indicating survival of virus $1$. Interestingly, aside
from the survival threshold $\tau_{1c}$, the winning threshold $\tau_{1}%
^{\dag}$ appears in the figure when plotted against a single virus case:
$\bar{p}_{1}^{ss}$ takes the same values as the single virus case for
effective infection rates larger than the winning threshold $\tau_{1}^{\dag}$.
For example FIG. \ref{ssinfectionpositive.eps} shows $\bar{p}_{1}^{ss}$ in the
competitive scenario (red curve) is exactly similar to the case of single
virus propagation (black curve) for $\tau_{1}>\tau_{1}^{\dag}\simeq6.6\frac
{1}{\lambda_{1}(A)}$. Hence, this region is identified as the absolute winning
range for virus $1$. For $\tau_{1}\in(\tau_{1c},\tau_{1}^{\dag})$, virus $1$
and virus $2$ each persist in the population, marking this range as the
coexistence region.

\begin{figure}[ptb]
\centering
\includegraphics[ width=3.4 in]{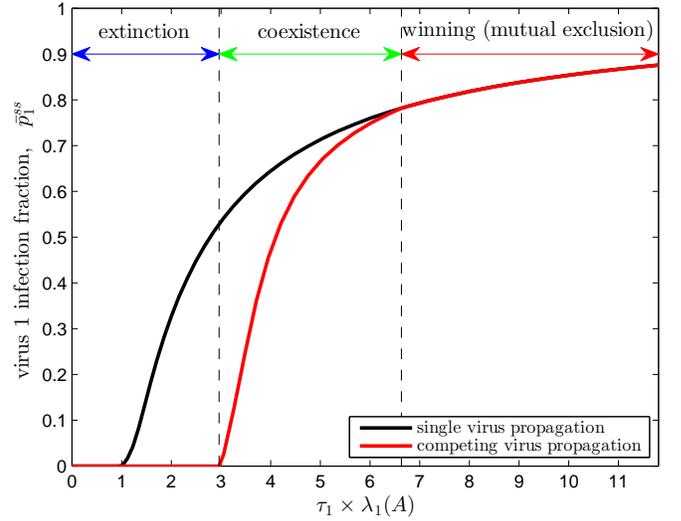}\caption{\emph{Steady
state infection fraction curve of virus }$1$\emph{\ in the }$SI_{1}SI_{2}%
S$\emph{\ competing spreading model (red).} While increasing $\tau_{1}$,
steady state infection fraction of virus $1$ in the the $SI_{1}SI_{2}S$ model
becomes nonzero at the survival threshold $\tau_{1c}$, while it coincides with
that of the $SIS$ model (black curve) at the winning threshold $\tau_{1}%
^{\dag}$. In this simulation, the steady-state infection fraction of virus $1$
($\bar{p}_{1}^{ss}$) is zero for $\tau_{1}\leq\tau_{1c}\simeq3\frac{1}%
{\lambda_{1}(A)}$, an extinction region for virus $1$. Interestingly, for
$\tau_{1}>\tau_{1}^{\dag}\simeq6.6\frac{1}{\lambda_{1}(A)}$, $\bar{p}_{1}%
^{ss}$ for the competitive scenario (red curve) is identical to the case of
single virus propagation (black curve), suggesting extinction of virus $2$,
hence marking this region as the winning range for virus $1$. For $\tau_{1}%
\in(\tau_{1c},\tau_{1}^{\dag})$, virus $1$ and virus $2$ both persist in the
population, marking this range for coexistence region.}%
\label{ssinfectionpositive.eps}%
\end{figure}

FIG. \ref{ssinfectionall.eps} illustrates the dependency of steady-state
infection fraction curve on network layer interrelation. When the contact
network of virus $2$ ($G_{B}$) is positively correlated with that of virus $1$
($G_{A}$), it is more difficult for virus $1$ to survive, making the survival
threshold $\tau_{1c}$ relatively larger for positively correlated $G_{B}$.
Negatively correlated contact network layers impede virus $1$ from completely
suppressing virus $2$, making winning threshold $\tau_{1}^{\dag}$ larger for
negatively correlated $G_{B}$.

\begin{figure}[ptb]
\centering
\includegraphics[ width=3.4 in]{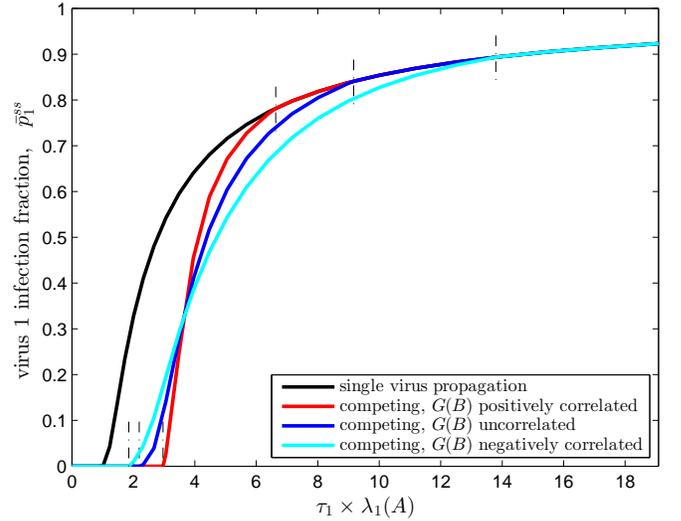}\caption{\emph{Comparison
of steady-state infection fraction curves of virus }$1$\emph{\ in the }%
$SI_{1}SI_{2}S$\emph{\ competitive spreading model.} Survival threshold
$\tau_{1c}$ is larger for positively correlated $G_{B}$, indicating it is more
difficult to survive positively correlated $G_{B}$, while $\tau_{1}^{\dag}$ is
larger for negatively correlated $G_{B}$, indicating it is more difficult to
completely suppress the other virus in negatively correlated $G_{B}$.}%
\label{ssinfectionall.eps}%
\end{figure}

\textbf{Survival diagram:} Allowing variation of $\tau_{2}$, the steady-state
infection curve extends to the steady-state infection surface. FIG.
\ref{ineffectionmap.eps} plots steady-state infection fraction for virus $1$
and virus $2$ as a function of $\tau_{1}$ and $\tau_{2}$. White curves
represent theoretical threshold curves derived from the solution to
(\ref{Thresh_Eq}), accurately separating the survival regions.

\begin{figure}[ptb]
\centering
\includegraphics[ width=3.4 in]{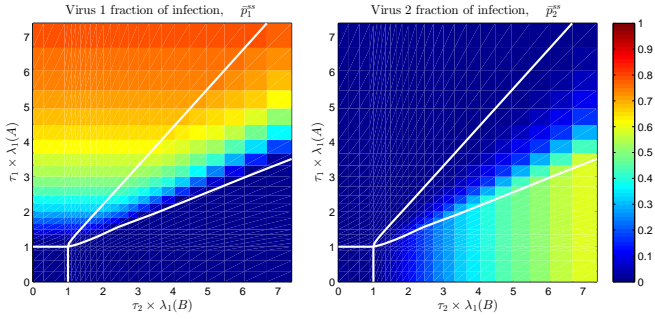}\caption{\emph{Steady state
fraction of infection for virus }$1$\emph{\ (left) and virus }$2$%
\emph{\ (right) as a function of }$\tau_{1}$\emph{\ and }$\tau_{2}$\emph{.}
The white lines are theoretical threshold curves accurately separating the
survival regions.}%
\label{ineffectionmap.eps}%
\end{figure}

FIG. \ref{thresholddiagramstwo.eps} plots standardized threshold diagram where
$G_{B}$ is negatively correlated with $G_{A}$ (left) and $G_{B}$ is positively
correlated with $G_{A}$ (right). Predictions from analytical approximation
formula (\ref{Bezier}) find the threshold curves fairly accurately.

\begin{figure}[ptb]
\centering
\includegraphics[ width=3.4 in]{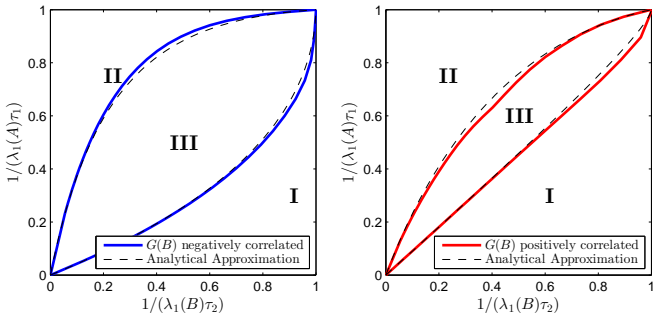}\caption{Standardized
threshold diagram for case where $G_{B}$ is negatively correlated with $G_{A}$
(left) and the case where $G_{B} $ is positively correlated with $G_{A}$
(right). Dashed lines are the predictions from analytical approximation
formula explicitly expressed in (\ref{Bezier}). Standardized threshold diagram
shows three survival regions: mutual extinction region I, where only virus $1$
survives and virus $2$ dies out, mutual extinction region II, where only virus
$2$ survives and virus $1$ dies out, and finally coexistence region III, where
both viruses survive and persist in the population.}%
\label{thresholddiagramstwo.eps}%
\end{figure}

\section{Discussion and Conclusion}

Competitive multi-virus propagation shows very rich behaviors, beyond those of
single virus propagation. This type of modeling is suitable for co-propagation
of exclusive entities, for example, opposing opinions about a subject, where
people are for, against, or neutral; spreading of a disease through physical
contact and viral propagation of antidote providing absolute immunity to the
disease, or marketing penetration of competitive products like Android versus
Apple smart phones. Aside from its potential applications, the problem of
competitive spreading over multilayer networks is technically challenging. In
particular, compared to single layer networks, science of multilayer networks
is still in its infancy. There are yet numerous unknowns about this complex problem.

In this paper, we study $SI_{1}SI_{2}S$ model, the simplest extension of $SIS$
model to competitive spreading over a two-layer network, focusing on long-term
behaviors in relation to multilayer network topology. In brief, the major
contributions of this paper are: (a) identification and quantification of
extinction, coexistence, and mutual exclusion via defining survival thresholds
and winning thresholds, (b) proving a region of coexistence and quantitating
it through overlapping of layers central nodes, (c) developing an explicit
approximation formula to globally find threshold values, and (d) proposing a
novel multilayer network generation scheme to capture influence of layers
interrelation. We believe our methodology has great potentials for application
to broader classes of multi-pathogen spreading over multi-layer and
interconnected networks.

\textbf{Acknowledgement:} This work was partially supported by National
Science Foundation under Award DMS-1201427. Any opinions, findings, and
conclusions or recommendations expressed in this paper are those of the
authors and do not necessarily reflect the views of the National Science Foundation.

\appendix*

\section{Selected Proofs}

\subsection{Derivation of Eigenvalue Perturbation Formulae}

Here, we detail the derivations of (\ref{dtaw1cdtaw2small}) and
(\ref{taw1ctaw2large}).

At $\tau_{2}=1/\lambda_{1}(B)$, (\ref{yi_eq}) finds $y_{i}=0$ for all nodes.
Equation (\ref{yi_eq}) is indeed the steady state equation for infection
probabilities in NIMFA model. Van Mieghem \cite{van2009TN} found for
SIS\ model the derivative with respect to effective infection rate, suggesting%
\begin{align}
\frac{dy_{i}}{d\tau_{2}}|_{\tau_{2}=\frac{1}{\lambda_{1}(B)}}  &
=c_{B}v_{B,i},\\
w_{i}|_{\tau_{2}=\frac{1}{\lambda_{1}(B)}}  &  =c_{A}v_{A,i}%
\end{align}
where%
\begin{equation}
c_{A}=\frac{\lambda_{1}(A)}{\sum v_{A,i}^{3}},~c_{B}=\frac{\lambda_{1}%
(B)}{\sum v_{B,i}^{3}},
\end{equation}
where $v_{A}$ and $v_{B}$ are the normalized dominant eigenvectors of $A$ and
$B$, respectively.

Differentiating (\ref{Thresh_Eq}) with respect to $\tau_{2}$ yields:
\begin{align}
\frac{dw_{i}}{d\tau_{2}}  &  =\frac{d\tau_{1c}}{d\tau_{2}}(1-y_{i})\sum
a_{ij}w_{j}\nonumber\\
&  +\tau_{1c}(-\frac{dy_{i}}{d\tau_{2}})\sum a_{ij}w_{j}\nonumber\\
&  +\tau_{1c}(1-y_{i})\sum a_{ij}\frac{dw_{j}}{d\tau_{2}}.
\end{align}
Inserting $\tau_{1c}=1/\lambda_{1}(A)$, $w_{i}=c_{A}v_{A,i}$, $y_{i}=0$, and
$dy_{i}/d\tau_{2}=c_{B}v_{B,i}$, the above equation changes to:%
\begin{equation}
(I-\frac{1}{\lambda_{1}(A)}A)\frac{d\boldsymbol{w}}{d\tau_{2}}=(\frac
{d\tau_{1c}}{d\tau_{2}})\lambda_{1}(A)c_{A}\boldsymbol{v}_{A}-c_{B}%
c_{A}(\boldsymbol{v}_{B}\circ\boldsymbol{v}_{A})
\end{equation}
in the collective form, where the Hadamard product $\circ$\ acts entry-wise.
Multiplying both sides by $\boldsymbol{v}_{A}^{T}$ from left yields:%
\begin{align}
\frac{d\tau_{1c}}{d\tau_{2}}|_{\tau_{2}=\frac{1}{\lambda_{1}(B)}}  &
=\frac{1}{\lambda_{1}(A)}c_{B}\boldsymbol{v}_{A}^{T}(\boldsymbol{v}_{B}%
\circ\boldsymbol{v}_{A})\nonumber\\
&  =\frac{\lambda_{1}(B)}{\lambda_{1}(A)}\frac{\sum v_{A,i}^{2}v_{B,i}}{\sum
v_{B,i}^{3}},
\end{align}
obtaining (\ref{dtaw1cdtaw2small}). Finding $\frac{d\tau_{1c}}{d\tau_{2}}$ at
$\tau_{2}=1/\lambda_{1}(B)$ obtains the dependence of $\tau_{1c}$ on $\tau
_{2}$ close to $1/\lambda_{1}(B)$.

Replacing for $1-y_{i}=\frac{\tau_{2}^{-1}}{\tau_{2}^{-1}+\sum b_{ij}y_{j}}$
from (\ref{yi_eq}) into (\ref{Thresh_Eq}) yields%
\begin{equation}
w_{i}=(\frac{\tau_{1c}}{\tau_{2}})(\frac{1}{\tau_{2}^{-1}+\sum b_{ij}y_{j}%
})\sum a_{ij}w_{j}.
\end{equation}
When effective infection rate $\tau_{2}$ is enormous $\tau_{2}^{-1}%
\rightarrow0$ and $y_{i}\rightarrow1$, suggesting%
\begin{equation}
w_{i}=(\frac{\tau_{1c}}{\tau_{2}}|_{\tau_{2}\rightarrow\infty})\frac
{1}{d_{B,i}}\sum a_{ij}w_{j},
\end{equation}
where $d_{B,i}$ is the $B-$degree of node $i$. Therefore, $\frac{\tau_{1c}%
}{\tau_{2}}|_{\tau_{2}\rightarrow\infty}$ is the inverse of the spectral
radius of $D_{B}^{-1}A$, proving (\ref{taw1ctaw2large}) for large values of
$\tau_{2}$.

\subsection{Coexistence Proofs}

\textbf{Coexistent region non-aggressive competitive viruses:}

To investigate the coexistence region for non-aggressive viruses we show that
(\ref{coexistence_1}) is true. From (\ref{dtaw1cdtaw2small}), we find%
\begin{multline}
\frac{d\tau_{1c}}{d\tau_{2}}\times\frac{d\tau_{2c}}{d\tau_{1}}|_{(\tau
_{1},\tau_{2})=(\frac{1}{\lambda_{1}(A)},\frac{1}{\lambda_{1}(B)})}\\
=\frac{(\sum v_{B,i}v_{A,i}^{2})(\sum v_{A,i}v_{B,i}^{2})}{(\sum v_{B,i}%
^{3})(\sum v_{A,i}^{3})} \label{Non-ageeIneq}%
\end{multline}

From H\"{o}lder's inequality%
\begin{align}
\sum v_{B,i}v_{A,i}^{2}  &  =\sum(v_{B,i}^{3})^{1/3}(v_{A,i}^{3}%
)^{2/3}\nonumber\\
&  \leq(\sum v_{B,i}^{3})^{1/3}(\sum v_{A,i}^{3})^{2/3}, \label{ineq1}%
\end{align}
and the equality happens iff $\boldsymbol{v}_{A}=\boldsymbol{v}_{B}$.
Similarly,%
\begin{equation}
\sum v_{A,i}v_{B,i}^{2}\leq(\sum v_{B,i}^{3})^{2/3}(\sum v_{A,i}^{3})^{1/3}.
\label{ineq2}%
\end{equation}
Multiplying sides of (\ref{ineq1}) and (\ref{ineq2}) yields%
\begin{equation}
(\sum v_{B,i}v_{A,i}^{2})(\sum v_{A,i}v_{B,i}^{2})\leq(\sum v_{B,i}^{3})(\sum
v_{A,i}^{3}),
\end{equation}
proving (\ref{Non-ageeIneq}) is true.

\textbf{Coexistent region for aggressive competitive viruses:}

To investigate the coexistence region for non-aggressive viruses we shown that
(\ref{coexistence_1}) is true. Substituting from (\ref{taw1ctaw2large}) yields%

\begin{align}
\frac{\tau_{1c}}{\tau_{2}}|_{\tau_{2}\rightarrow\infty}.\frac{\tau_{2c}}%
{\tau_{1}}|_{\tau_{1}\rightarrow\infty}  &  =\frac{1}{\lambda_{1}(D_{B}%
^{-1}A)}.\frac{1}{\lambda_{1}(D_{A}^{-1}B)}\nonumber\\
&  =\frac{1}{\lambda_{1}(D_{B}^{-1}A\otimes D_{A}^{-1}B)}\nonumber\\
&  =\frac{1}{\lambda_{1}[(D_{B}^{-1}\otimes D_{A}^{-1})(A\otimes
B)]}\nonumber\\
&  =\frac{1}{\lambda_{1}[(D_{B}\otimes D_{A})^{-1}(A\otimes B)]},
\end{align}
according to properties of Kronecker product.

The degree diagonal matrix of $(A\otimes B)$ is $(D_{A}\otimes D_{B})$.
Therefore, $(D_{B}\otimes D_{A})$ is a diagonal permutation of the degree
diagonal matrix of $(A\otimes B)$. According to Lemma 1, presented in the
following, $\lambda_{1}[(D_{B}\otimes D_{A})^{-1}(A\otimes B)]\geq1$, thus%
\begin{equation}
\frac{\tau_{1c}}{\tau_{2}}|_{\tau_{2}\rightarrow\infty}.\frac{\tau_{2c}}%
{\tau_{1}}|_{\tau_{1}\rightarrow\infty}\leq1,
\end{equation}
and equality holds only if $D_{B}\otimes D_{A}=D_{A}\otimes D_{B}$, which
holds only if ratio of $B-$degree and $A-$degree of each node is same for all nodes.

\begin{lemma}
If $H=\pi(D_{C})^{-1}C$, where $\pi(D_{C})$ is a diagonal permutation of
degree diagonal matrix of symmetric matrix $C$, then $\lambda_{1}(H)\geq1$.
Furthermore, equality holds only if $\pi(D_{C})=D_{c}$.
\end{lemma}

\begin{proof}
The largest eigenvalue maximizes Rayleigh quotient, therefore,%
\begin{align*}
\lambda_{1}(H)  &  =\lambda_{1}(\pi(D_{C})^{-1}C)=\lambda_{1}(\pi
(D_{C})^{-1/2}C\pi(D_{C})^{-1/2})\\
&  =\max_{x}\frac{x^{T}\pi(D_{C})^{-1/2}C\pi(D_{C})^{-1/2}x}{x^{T}x}\\
&  \geq\frac{1^{T}C1}{1^{T}\pi(D_{C})1}=\frac{\sum d_{C,i}}{\sum d_{C}%
,_{\pi_{i}}}=1,
\end{align*}
where $d_{C},_{\pi_{i}}$ is the degree of node $i$ map. Therefore,
$\lambda_{1}(H)\geq1$. Equality holds only if $x=\pi(D_{C})^{1/2}1$ is the
dominant eigenvector of $\pi(D_{C})^{-1/2}C\pi(D_{C})^{-1/2}$, i.e.,
$\pi(D_{C})^{-1/2}C1=\pi(D_{C})^{1/2}1$, which only holds if $d_{C},_{\pi_{i}%
}=d_{C,i}$.
\end{proof}

\subsection{Steady State Numerical Solution}

Given $\tau_{2}>1/\lambda_{1}(B)$, (\ref{Thresh_Eq}) and (\ref{yi_eq})
numerically find $\tau_{1,c}$. We now define $x_{i}\triangleq\frac{y_{i}%
}{1-y_{i}}$, given the recursive iteration law:%
\begin{equation}
x_{i}(k+1)=\tau_{2}\sum b_{ij}\frac{x_{j}(k)}{1+x_{j}(k)} \label{SIS_SS_Law}%
\end{equation}
to prove they converge exponentially, numerically solving (\ref{yi_eq}) as
$\frac{x_{i}(k)}{1+x_{i}(k)}\rightarrow y_{i}$. The main advantage of finding
equilibrium values using recursive law (\ref{SIS_SS_Law}) instead of solving
ordinary differential equations of the model is recursive law
(\ref{SIS_SS_Law}) does not require incremental time increase, making
computations drastically faster.

Furthermore, the steady-state infection probabilities in (\ref{EqEq1}%
)-(\ref{EqEq2}) can be found via the recursive iteration law:%
\begin{align}
x_{i}(k+1)  &  =\tau_{1}\sum a_{ij}\frac{x_{j}(k)}{1+x_{j}(k)+z_{j}(k)},\\
z_{i}(k+1)  &  =\tau_{2}\sum b_{ij}\frac{z_{j}(k)}{1+x_{j}(k)+z_{j}(k)},
\end{align}
as $\frac{x_{i}(k)}{1+x_{i}(k)+z_{i}(k)}\rightarrow p_{1,i}^{\ast}$ and
$\frac{z_{i}(k)}{1+x_{i}(k)+z_{i}(k)}\rightarrow p_{2,i}^{\ast}$.

\bibliographystyle{revcompchem}
\bibliography{MultiPathogene}

\end{document}